\documentclass[onefignum,onetabnum]{siamart190516}

\usepackage{lipsum}
\usepackage{amsfonts,amsmath,amssymb}
\usepackage{graphicx}
\usepackage{epstopdf}
\usepackage{algorithmic}
\usepackage{mathrsfs}
\usepackage{dsfont} 
\ifpdf
  \DeclareGraphicsExtensions{.eps,.pdf,.png,.jpg}
\else
  \DeclareGraphicsExtensions{.eps}
\fi

\newsiamremark{Remark}{Remark}
\newsiamremark{Corollary}{Corollary}
\newsiamremark{hypothesis}{Hypothesis}
\crefname{hypothesis}{Hypothesis}{Hypotheses}
\newsiamthm{claim}{Claim}

\headers{Quantum projection filter and stabilization}{Quantum projection filter and stabilization}

\title{Feedback stabilization via a quantum projection filter\thanks{Submitted to the editors DATE.
\funding{This work is supported by the ANR project  Q-COAST Projet- ANR-19-CE48-0003 and the ANR project IGNITION ANR-21-CE47-0015.}}}

\author{Nina H. Amini\thanks{Laboratoire des signaux et syst\`{e}mes (L2S), CNRS-CentraleSup\'{e}lec-Universit\'{e} Paris-Sud, Universit\'{e} Paris-Saclay, 3, rue Joliot Curie, 91190 Gif-sur-Yvette, France
\email{ (nina.amini@l2s.centralesupelec.fr}).}\and 
Paolo Mason\thanks{Laboratoire des signaux et syst\`{e}mes (L2S), CNRS-CentraleSup\'{e}lec-Universit\'{e} Paris-Sud, Universit\'{e} Paris-Saclay, 3, rue Joliot Curie, 91190 Gif-sur-Yvette, France \email{ (paolo.mason@l2s.centralesupelec.fr}).}
\and Ibrahim Ramadan\thanks{Laboratoire des signaux et syst\`{e}mes (L2S), CNRS-CentraleSup\'{e}lec-Universit\'{e} Paris-Sud, Universit\'{e} Paris-Saclay, 3, rue Joliot Curie, 91190 Gif-sur-Yvette, France \email{ (ibrahim.ramadan@l2s.centralesupelec.fr}).}
}

\usepackage{amsopn}

\ifpdf
\hypersetup{
  pdftitle={Projectionibr},
  pdfauthor={}
}
\fi

\begin{document}

\maketitle
\begin{abstract}
This paper considers a simplified model of open quantum systems undergoing imperfect measurements obtained via a projection filter approach. We use this { approximate filter} in the feedback stabilization problem specifically in the case of Quantum Non-Demolition (QND) measurements. 
The feedback design relies on the structure of the exponential family utilized for the projection process. We demonstrate that the introduced feedback guarantees exponential convergence of the original filter equation toward a predefined target state, corresponding to an eigenstate of the measurement operator.
\end{abstract}




\section{Introduction}
{ Controlling quantum dynamics in a reliable way} represents a crucial milestone in the advancement of quantum technologies. Quantum systems can be categorized as closed or open. Unlike closed quantum systems, which are supposed to remain isolated from the environment, open quantum systems interact with their surroundings, representing a more realistic scenario for physical systems. However, the interaction with the environment leads to decoherence phenomena, causing information loss, see e.g.,  \cite{davies1976quantum}. The dynamics of open quantum systems are captured by quantum Langevin equations, which can be derived by using quantum stochastic calculus \cite{Hudson1984QuantumIF} 
and the input-output formalism \cite{gardiner2004input}.


 Quantum measurements exhibit a probabilistic nature, introducing random back-action on the system, a property that lacks a classical analog. The conditional evolution of the quantum system state is described by a stochastic master equation, which is derived from quantum filtering theory, originally developed by Belavkin \cite{belavkin1992quantum,belavkin1989nondemolition}, see also \cite{bouten2007introduction} for a recent reformulation. In the physics and probability community, the term quantum trajectory is more commonly used to indicate the stochastic evolution of the filter state, see e.g., \cite{barchielli2009quantum,wiseman2009quantum,Carmichael1993AnOS}. 

The control of open quantum systems represents a fundamental area of study to either suppress decoherence or exploit it efficiently to achieve specific objectives \cite{belavkin2004towards}. Because of the need for increased robustness, closed-loop control strategies 
may be preferable, compared to
open-loop control strategies. Feedback strategies can be implemented whether or not measurements are present.
In the presence of measurements, feedback is based on partial information obtained from the measurement outcomes. This approach is commonly known as measurement-based feedback control.

In practical experiments, numerous factors contribute to imperfections, including uncertain initial states and detector inefficiency, see e.g., \cite{sayrin2011real}. As a result, ensuring the robustness of control strategies becomes a crucial concern. This issue is addressed in various papers, see e.g., \cite{liang2020robustness,lidar1998decoherence,ticozzi2008quantum}. Another issue is represented by the time-delays that may appear in the real-time implementation of  measurement-based feedbacks and may impact their effectiveness.
This is  due to the relatively slow processing time of {classical measurements by digital electronics} and the high dimensions of quantum filters. 
In \cite{Cardona2019ExponentialSO}, the authors established a feedback stabilization result for systems undergoing QND measurements and proposed  an approximate filter which can be used for the feedback design. The efficiency of the latter was shown through numerical studies. In \cite{liang2022model}, the authors gave a formal proof of the effectiveness of such a feedback design.


To reduce the representation complexity of quantum filters, a projection filter methodology has been introduced for quantum systems in \cite{Vanhandel1}. This generalizes the projection filter methodology considered for classical filter equations. The idea is to use differential and information geometry tools, as stated in~\cite{Amari,Brigo1,Brigo2}.
Subsequently, { unsupervised learning techniques, specifically local tangent space alignment, were employed in~\cite{Nielsen} in order} to recast the system state's evolution within a lower-dimensional manifold. The work presented in \cite{tezak2017low} derived a dynamical law by minimizing the statistical distance in the moving basis, establishing its equivalence to the projection filter approach. A quantum projection filtering approach was developed in \cite{Gao}, projecting the system's dynamics onto the tangent space of  an exponential family of unnormalized density matrices. 
Recently in \cite{ramadan2022exact}, we obtained an exact solution of quantum filters for QND measurement expressing the quantum trajectory
 in terms of the solution of a lower dimensional stochastic differential equation. Also we performed an error analysis for  a quantum projection filter based on an exponential family in the case of imperfect QND measurement; this extends the result
of \cite{Gao} which treats the perfect QND measurement. 

The projection filter approach may be particularly interesting in order to design feedback controls  implementable in real-time, see e.g., \cite{Rouchon} where numerical studies have been realized concerning the use of an approximate filter in feedback design.   Numerous studies focused on the control of open quantum systems to achieve the preparation of pure states. These investigations are crucial for advancing quantum technologies. The pioneering work \cite{Vanhandel2} addressed the problem of designing a quantum feedback controller that globally stabilizes a quantum spin-$\frac{1}{2}$ system toward an eigenstate of the measurement operator $\sigma_z$, even when imperfect measurements are involved. The controller was developed using numerical methods to find an appropriate global Lyapunov function. Subsequently, in \cite{Vanhandel3}, the authors employed stochastic Lyapunov techniques to analyze the stochastic flow and constructed a switching control law to stabilize $N$-level quantum angular momentum system toward an eigenstate of the measurement operator.  The recent works \cite{liang2019exponential,Liang} combined local stochastic stability analysis with the support theorem to establish exponential stabilization results for spin-$\frac{1}{2}$ and spin-$J$ systems toward stationary states of the open-loop dynamics. Unlike previous approaches based on the LaSalle method, their techniques allowed  to estimate the convergence rate to the target state. This estimation is crucial for practical implementation in quantum information processing. In \cite{Cardona2018ExponentialSS,Cardona2019ExponentialSO}, the authors present an exponential stabilization using a noise-assisted feedback. We also refer to \cite{amini2023exponential}, where an exponential feedback stabilization in expectation toward invariant subspaces of the evolution for generic measurement has been developed.

In this paper, we apply a projection filter approach  for open quantum systems undergoing indirect measurement, in presence of detection imperfections and unknown initial state. To define a projection filter, we make use, as in \cite{Gao,ramadan2022exact}, of a parametrized  exponential family. We then tackle the feedback stabilization problem for a system describing the evolution of a $N-$level quantum angular momentum system undergoing imperfect QND measurements, assuming that the feedback depends on the projection filter.  In this case, we deal with a coupled stochastic equation consisting of the quantum filter and its corresponding projection filter with an arbitrary fixed initial state. To tackle the problem of feedback stabilization in this case, we first show that the analysis done in \cite{liang2021robust} can be adapted to such framework. Then, in view of reducing the computational complexity of the real-time implementation of the feedback, we introduce a new lower-dimensional parametrization of the exponential family, so that the original problem can be reformulated as a stabilization problem for a coupled system describing the evolution of a pair formed by the actual filter and a vector of parameters. We then seek for a feedback as a function of this new parametrization. We provide sufficient conditions on the feedback controller ensuring the exponential stabilization of the coupled system toward a chosen target state. Furthermore, we propose some examples of feedback which satisfy such conditions. 


This paper is structured as follows. In Section~\ref{A}, we present the quantum filter equation and its estimation
assuming that the initial state is unknown.
In Section~\ref{B}, we describe the considered projection filter approach.
In Section~\ref{C}, we consider $N$-level quantum angular momentum systems and the feedback stabilization problem based on the projection filter. We give sufficient conditions on the feedback controller stabilizing the filter equation toward a chosen eigenstate of the measurement operator. Section~\ref{D} focuses on the feedback stabilization problem based on a parametrization of the exponential family used for the projection filter.   
Numerical simulations are provided in Section~\ref{E} for four-level quantum angular momentum systems. Section~\ref{F} concludes this work and gives further research lines.

{\bf{Notation.}} 
The complex conjugate transpose of a matrix $A$ is denoted by $A^\dag$. The commutator of two square matrices $A$ and $B$ is represented as $[A, B]=AB-BA$. Assuming that $X$ is a metric space, we indicate by $B_\varepsilon(x)$ the ball of radius $\varepsilon$ around $x\in X$.

\section{Problem Description}\label{A}
Given a $N$-dimensional quantum system, its state can be described by the density operator $\rho$, which is 
a { positive semi-definite} Hermitian matrix of trace one, i.e., $\rho$ belongs to the set
\begin{equation*}\mathcal{S}_N:=\left\{\rho \in \mathbb C^{N \times N}\mid \rho=\rho^{\dagger}, \rho \geq 0, \operatorname{Tr}(\rho)=1\right\}.\end{equation*}
The state evolution of the quantum system undergoing continuous-time measurement can be described by the following quantum {stochastic master equation}  (see e.g., \cite{belavkin1989nondemolition})
\begin{equation}
\begin{array}{l}
d \rho_t=-i\left[H, \rho_t\right] d t+\mathcal{F}( \rho_t )d t+\mathcal{G}( \rho_t) d W_t \\
d Y_t=d W_t+\sqrt{\eta}\operatorname{Tr}\left(\rho_t(L+L^{\dagger})\right) d t,
\end{array}
 \label{belavkin}
\end{equation}
where $L$ corresponds to the measurement operator, $Y$ represents the observation process in the case of homodyne detection, and the super-operators $\mathcal{F}$ and $\mathcal{G}$ have the following form
\begin{equation*}
\begin{array}{l}
\mathcal{F}( \rho):=L \rho L^{\dagger}-\frac{1}{2}\left(L^{\dagger} L \rho+\rho L^{\dagger} L\right) \\
\mathcal{G}( \rho):=\sqrt{\eta}\left(L \rho+\rho L^{\dagger}-\operatorname{Tr}\left[(L+L^{\dagger}) \rho\right] \rho\right).
\end{array}
\end{equation*} 
The operator $H$ represents the Hamiltonian of the system, which is formed by a free Hamiltonian and a controlled term, 
i.e., $H=H_0+ u H_1$ where $u$ 
is the control input. Note that, due to the dependence on the control input, the Hamiltonian is in general time-varying. We denote it simply by $H$ for the sake of simplicity.  
In the following the control input will be chosen in the form of a state-feedback, i.e., 
$u=u(\rho).$
In the equation~\eqref{belavkin} { $W_t$} is a classical Wiener process with respect to  { a probability space $(\Omega,\mathcal F, \mathbb P$)}. The efficiency of the detector is expressed by the parameter $\eta \in(0,1]$. 

In practice, we do not have { the quantum state at our disposal}. In this case an estimation of the above quantum filter is considered. A natural choice is to consider an estimate denoted by $\hat\rho_t$ following the same dynamics as above with an arbitrary initial state $\hat\rho_0,$ 
\begin{align}
 \label{estimate}
d \hat \rho_t & =-i\left[H, \hat \rho_t\right] d t+\mathcal{F}( \hat \rho_t )d t+\mathcal{G}( \hat \rho_t)  (d Y_t-\sqrt{\eta}\operatorname{Tr}\left(\hat \rho_t(L+L^{\dagger})\right) d t)\\
& = \left(-i\left[H, \hat \rho_t\right] +\mathcal{F}( \hat \rho_t )+\sqrt{\eta}\mathcal{G}(\hat \rho_t)\operatorname{Tr}\left((\rho_t-\hat \rho_t)(L+L^{\dagger})\right)\right) d t+\mathcal{G}( \hat \rho_t)  d W_t.\nonumber
\end{align}
The stochastic master equation~\eqref{belavkin} (coupled or not with the equation~\eqref{estimate}) has been extensively explored in various fields such as quantum state estimation and quantum feedback control \cite{Vanhandel2,van2005modelling,liang2021robust}. 
The evolution of such an It\^o stochastic differential equation takes place in a $N^2-1$ dimensional space. 
For large values of $N$ the time required to solve the equation may become 
extremely high. This represents an obstacle to the use of a feedback control for stabilization based on a real time evaluation of $\hat \rho_t$. 
Therefore, it is common practice to use an 
estimate of the system's state based on a lower dimension model of the dynamics, instead of directly computing the solution of~\eqref{estimate}. This motivates us to use a model reduction method called the {quantum projection filter} approach.
\section{Quantum projection filter}\label{B}
The concept of quantum projection filtering involves constructing a low-dimensional submanifold of the Hilbert space and a corresponding geometric projection allowing to map the 
dynamics
onto the submanifold. This approach can be used to obtain an estimate of the quantum state~\cite{Vanhandel1}.
The so-called exponential quantum projection filter represents an example of application of this idea~\cite{Gao,ramadan2022exact}. In this case, one approximates an unnormalized version of the system state, i.e., a matrix $\check \rho_t$ satisfying ${\rho}_t = \check \rho_t/{\rm Tr}(\check{\rho}_t)$, which is the solution of the following Zakai equation, 
\begin{equation}
d \check{\rho}_{t}=\left(- i\left[H, \check{\rho}_{t}\right]+ 
\mathcal{F}(\check\rho_t)\right) d t+\sqrt{\eta}\left(L \check{\rho}_{t}+\check{\rho}_{t} L^{\dagger}\right) d Y_{t}. \label{2}
\end{equation}
In the following, we also consider the corresponding Stratonovich form 
\begin{equation}
d {\check{\rho}}_{t}=\left(-i\left[H, {\check{\rho}}_{t}\right]+\hat{\mathcal{F}}\left({\check{\rho}}_{t}\right)\right) d t+\sqrt{\eta}\left(L {\check{\rho}}_{t}+{\check{\rho}}_{t} L^{\dagger}\right) \circ d Y_{t}, \label{stra}
\end{equation} 
where 
$\hat{\mathcal{F}} \left(\rho\right)=(1-\eta)L \rho L^{\dagger}-\frac{\left(\eta L+L^{\dagger}\right) L \rho+\rho L^{\dagger}\left(L+\eta L^{\dagger}\right)}{2}.$ Compared to the corresponding It\^o form~\eqref{2}, the equation~\eqref{stra} is more compatible with the manifold structure of the state space (see, e.g, \cite{elworthy_1982}).

The aim is to obtain an approximate solution $\check{\rho}_{\theta_t}$ of the above  equation. A natural choice is to consider the following parametrized exponential family (see e.g., \cite{Gao})
\begin{equation}
\check{\mathcal{M}}=\left\{\check{\rho}_{\theta}:=e^{\frac{1}{2} \sum_{j=0}^{m-1} \theta_j A_j} \bar\rho_0 e^{\frac{1}{2} \sum_{j=0}^{m-1} \theta_j A_j}\mid \theta\in \mathbb{R}^m\right\},
\label{manifold}
\end{equation}
where $\bar\rho_0$ represents the initial condition of the approximated dynamics and can be chosen arbitrarily. The operators $A_j$ for $j=0,1, \ldots, m-1$ are assumed to be {Hermitian $N \times N$ matrices}.
Locally, $\check{\mathcal{M}}$ is a $m$-dimensional submanifold of the space of positive semi-definite Hermitian operators in $\mathcal{H}$, provided that the set $\{\frac{\partial{\check{\rho}_{\theta}}}{\partial \theta_{0}}, \dots, \frac{\partial{\check{\rho}_{\theta}}}{\partial \theta_{m-1}}\}$
 is linearly
independent.

The objective is to project the dynamics~\eqref{stra} onto $\check{\mathcal{M}}$  and to deduce the corresponding dynamics of the parameter $\theta.$ The latter can be obtained by noting  that the chain rule for Stratonovich stochastic calculus yields
\begin{equation}
d \check\rho_\theta=\sum_{j=0}^{m-1} \check{\partial}_j \circ d \theta_j,
\label{chain}
\end{equation}
where $\check{\partial}_j:=\frac{\partial \check{\rho}_\theta}{\partial \theta_j}.$
If the operators $A_j$  are mutually commuting, the above derivative can be explicitly computed as 
\begin{align*}
\check{\partial}_j=\frac{1}{2}\left(A_j \check\rho_\theta+\check\rho_\theta A_j\right)
\end{align*}

The choice of $m$ and of the operators $A_j$ should be done in a way that the computational complexity of the projected dynamics is significantly reduced compared to the original one, while keeping the dynamics of $\check\rho_t$ and $\check\rho_{\theta_t}$ as close as possible.

From now on, we assume that $L$ is a {Hermitian $N \times N$ matrix} 
so that  by spectral decomposition, we can write $L=\sum_{j=1}^{n_{0}} \lambda_j \pi_j$, where $\lambda_j$  are the real distinct eigenvalues of $L$ and $\pi_j$ are the corresponding orthogonal projection operators{, for $j=1, \ldots, n_0$, with $n_0\leq N$}. 
In the following, we take
\begin{equation}
\begin{array}{l}
m=n_0, \quad 
A_j=\pi_{j+1},
\end{array}
\label{Assumption}
\end{equation}
in the definition of $\check{\mathcal{M}}.$ 
Furthermore, we assume from now on that $\bar\rho_0$ is chosen so that $\mathrm{Tr}(\bar \rho_0 A_k)\neq 0$ for every $k=0,\ldots,{m}-1$. 
{This ensures that $\mathrm{Tr}(\check \rho_\theta A_k)\neq 0$, as $ \mathrm{Tr}(\check \rho_\theta A_k)=\mathrm{Tr}(\sum_{i,j=0}^{m-1}e^{\frac{\theta_i+\theta_j}{2}}A_i\bar{\rho}_0A_j A_k)=e^{\theta_k}\mathrm{Tr}(\bar \rho_0 A_k)$ for every $\theta\in\mathbb{R}^m$.}

As it is explained in Appendix \ref{Orthogonal projection}, the projection can be defined 
via a map $\Pi_\theta$ from the space of Hermitian matrices to the tangent space of $\check{\mathcal{M}}.$ The projection filter can be obtained by applying such a map to each term of the dynamics of $\check{\rho}$ as follows 
\begin{align}
d \check{\rho}_{\theta_{t}}=\Pi_{\theta_{t}}\left(-i\left[H, \check{\rho}_{\theta_{t}}\right]\right)d t+\Pi_{\theta_{t}}\left({\hat{\mathcal{F}}} \left(\check{\rho}_{\theta_{t}}\right)\right) d t+\Pi_{\theta_{t}}\left(\sqrt{\eta}(L \check{\rho}_{\theta_{t}}+\check{\rho}_{\theta_{t}} L)\right) \circ d Y_{t}.\label{Projection}
\end{align}
From the definition of such a map, we deduce that 
\begin{equation}
\begin{aligned}
d \check{\rho}_{\theta_t}=&\sum_{j=0}^{m-1} \frac{1}{2}\left(A_j \check{\rho}_{\theta_t}+\check{\rho}_{\theta_t} A_j\right)\left(\frac{\operatorname{Tr}\left({i} \check{\rho}_{\theta_t}\left[H, A_j\right]\right)}{\operatorname{Tr}\left(\check{\rho}_{\theta_t} A_j\right)}\right)  d t- \\
&\eta (L^2\check{\rho}_{\theta_t}+\check{\rho}_{\theta_t}L^2) d t+\sqrt{\eta} (L\check{\rho}_{\theta_t}+\check{\rho}_{\theta_t}L) \circ d Y_t \\
=& \sum_{j=0}^{m-1} \frac{1}{2}\left(A_j \check{\rho}_{\theta_t}+\check{\rho}_{\theta_t} A_j\right)\bigg(\frac{\operatorname{Tr}\left({i} \check{\rho}_{\theta_t}\left[H, A_j\right]\right)}{\operatorname{Tr}\left(\check{\rho}_{\theta_t} A_j\right)} d t- 
2\eta \lambda_{j+1}^2 d t+2\sqrt{\eta} \lambda_{j+1} \circ d Y_t\bigg).
\label{unnormalized stra}
\end{aligned}
\end{equation}
Now it is sufficient to use the relation~\eqref{chain} {and~\eqref{unnormalized stra}} to get {the following equation in It\^o form}
\begin{equation}
d \theta_t=G\left(\theta_t\right)^{-1} E\left(\theta_t\right) d t-2 \eta \alpha d t+2 \sqrt{\eta}\beta d Y_t,
\label{Quantum projection}
\end{equation}
with 
$$
\begin{gathered}
G\left(\theta_t\right)=\operatorname{diag}\left\{\operatorname{Tr}\left(\check{\rho}_{\theta_t} A_1\right), \ldots, \operatorname{Tr}\left(\check{\rho}_{\theta_t} A_m\right)\right\}, \\
E\left(\theta_t\right)=\left(\operatorname{Tr}\left(i \check{\rho}_{\theta_t}\left[H, A_1\right]\right), \ldots, \operatorname{Tr}\left(i \check{\rho}_{\theta_t}\left[H, A_m\right]\right)\right)^{T} \\
\alpha=\left(\lambda_1^2, \ldots, \lambda_m^2\right)^{T},~\beta=\left(\lambda_1, \ldots, \lambda_m\right)^{T} .
\end{gathered}
$$
{Note that $G\left(\theta_t\right)$ is invertible since  $\mathrm{Tr}(\check \rho_\theta A_k)\neq 0$ for $k=0,\ldots, m-1,$ as explained above.}\\
The unnormalized It\^{o} stochastic differential equation takes the form
\begin{equation}
\begin{aligned}
d \check{\rho}_{\theta_t}=&\frac{1}{2}\sum_{j=0}^{m-1} \left(A_j \check{\rho}_{\theta_t}+\check{\rho}_{\theta_t} A_j\right)\left(\frac{\operatorname{Tr}\left({i} \check{\rho}_{\theta_t}\left[H, A_j\right]\right)}{\operatorname{Tr}\left(\check{\rho}_{\theta_t} A_j\right)}\right)  d t+  \\
&\eta \left(L\check{\rho}_{\theta_t}L-\frac{1}{2}(L^2\check{\rho}_{\theta_t}+\check{\rho}_{\theta_t}L^2)\right) d t+\sqrt{\eta} (L\check{\rho}_{\theta_t}+\check{\rho}_{\theta_t}L) d Y_t. \label{unnormalized Ito}
\end{aligned}
\end{equation}
By applying It\^{o} rules we can obtain the equation for the normalized solution of the projection filter ${\rho}_{\theta_t}:=\frac{\check{\rho}_{\theta_t}}{\operatorname{Tr}\left(\check{\rho}_{\theta_t}\right)}$, which is given as follows 
\begin{equation}
\begin{aligned}
d{\rho}_{\theta_t}
=&\frac{1}{2}\sum_{j=0}^{m-1} \left(A_j {\rho}_{\theta_t}+{\rho}_{\theta_t}A_j\right)\left(\frac{\operatorname{Tr}\left({i} {\rho}_{\theta_t}\left[H, A_j\right]\right)}{\operatorname{Tr}\left({\rho}_{\theta_t} A_j\right)}\right)  d t +  \\
&\eta\left(L{\rho}_{\theta_t} L-\frac{1}{2}(L^2{\rho}_{\theta_t}+{\rho}_{\theta_t} L^2)\right) d t+\mathcal{G}( {\rho}_{\theta_t}) 
(d Y_t-2\sqrt{\eta}\operatorname{Tr}\left({\rho}_{\theta_t}L\right) d t).
\label{normalized Ito}
\end{aligned}
\end{equation}
The evolution of the equation~\eqref{normalized Ito} starting from $\bar\rho_0$ takes place almost surely on the set \[\mathcal{M}=\left\{\frac{\rho}{\mathrm{Tr}(\rho)}\mid\rho\in\check{\mathcal{M}}\right\}\subset \mathcal{S}_N.\] 

In the following, we will study a feedback stabilization problem based on the projection filter given in~\eqref{normalized Ito}. 

\section{Quantum projection filter approach  for feedback stabilization}\label{C}
In this section, we consider a $N$-level quantum angular momentum system undergoing continuous-time non-demolition measurements. This is a typical example considered in the literature, see e.g., \cite{Vanhandel3}. The aim is to study feedback stabilization of such a system with a feedback which depends on the projection filter. 

First let us introduce the mathematical model. 
The evolution of such a physical system is described by the following stochastic master equation 
\begin{equation}
\begin{array}{l}
d \rho_t=-i\left[\omega J_z+uJ_y, \rho_t\right] d t+\mathcal{F}( \rho_t )d t+\mathcal{G}( \rho_t) d W_t \\
d Y_t=d W_t+2\sqrt{\eta}\operatorname{Tr}\left(J_z \rho_t)\right) d t
\end{array},
 \label{Optimal Filter}
\end{equation}
where
\begin{itemize}
\item The measurement operator is denoted by $J_z$ which is the self-adjoint angular momentum along the $z$-axis. The matrix representation of $J_z$ is
$$
J_z=\left[\begin{array}{lllll}
J & & & & \\
& J-1 & & & \\
& & \ddots & & \\
& & & -J+1 & \\
& & & & -J
\end{array}\right],
$$
with  $J:=\frac{N-1}{2}.$
\item Similarly, $J_y$ represents the self-adjoint angular momentum along the $y$-axis. It is defined as follows:
$$
J_y=\left[\begin{array}{ccccc}
0 & -i c_1 & & & \\
i c_1 & 0 & -i c_2 & & \\
& \ddots & \ddots & \ddots & \\
& & i c_{2 J-1} & 0 & -i c_{2 J} \\
& & & i c_{2 J} & 0
\end{array}\right],
$$
with $c_q=\frac{1}{2} \sqrt{(2 J+1-q) q}$ for $q=1,\ldots,2J.$
\item We suppose that the feedback $u$ is  adapted to the algebra generated by the observation process up to time $t$, for every $t>0.$

\item $\omega>0$ is a physical parameter.

\end{itemize}

\medskip

The following describes the time evolution of the pair ($\rho_t,{\rho}_{\theta_t}$) $\in \mathcal{S_N} \times \mathcal{S_N}$,
\begin{equation}
d \rho_t=-i\left[\omega J_z+uJ_y, \rho_t\right] d t+\mathcal{F}( \rho_t )d t+\mathcal{G}( \rho_t) d W_t \label{coupled1}
\end{equation}

\begin{equation}
    \begin{aligned}
d {\rho}_{\theta_t}=&\frac{1}{2}\sum_{j=0}^{2J} \left(A_j {\rho}_{\theta_t}+{\rho}_{\theta_t} A_j\right)\left(\frac{\operatorname{Tr}\left({i} {\rho}_{\theta_t}\left[uJ_y, A_j\right]\right)}{\operatorname{Tr}\left({\rho}_{\theta_t} A_j\right)}\right)  d t+  \\
&{\eta} \mathcal{F}({\rho}_{\theta_t} )
d t  +2\sqrt{{\eta}} \mathcal{G}({\rho}_{\theta_t})\left(Tr(J_z \rho_t)-Tr(J_z {\rho}_{\theta_t})\right)d t +\mathcal{G}( {\rho}_{\theta_t}) d {W_t}
\label{coupled2}
\end{aligned}
\end{equation}
The aim of this section is to study feedback stabilization of the original filter $\rho$ toward one of its equilibria $A_{\bar n}$ with $\bar n\in\{0,\ldots,2J\}.$ Here we note that $A_k$ corresponds to the rank-one projector associated with the eigenvalue $J-k$ of $J_z$ for $k=0,\ldots,2J.$ It is easy to see that $A_{k}$ is in $\mathrm{clos}(\mathcal{M})$ for $k=0,\ldots,m-1,$ where  $\mathrm{clos}(\cdot)$ indicates the topological closure. The key point is that the feedback $u$ will be assumed to depend on $\rho_\theta$. 
In this case the above equations are coupled, i.e., the evolution of $\rho$ depends on $\rho_\theta$ through $u$ and, vice versa, the evolution of $\rho_\theta$ depends on $\rho$ through the measurement process $Y$.

In absence of feedback, in many papers, see e.g., \cite{adler2001martingale,van2005modelling,Vanhandel3,benoist2014large}, it was shown that the quantum filter~\eqref{coupled1} randomly selects  one of the equilibria $A_j.$ This phenomenon is called quantum state reduction. 
In the following theorem, we show that the same asymptotic behavior holds true for the projection filter~\eqref{coupled2} in absence of feedback. 

Here we apply Definition \ref{stability} for the notion of stochastic stability and, for the  corresponding distance, we consider the Bures distance defined as follows. 
\begin{definition}[Bures distance, see \cite{bengtsson2017geometry}]
\begin{itemize}
\item The Bures distance between two density matrices $\rho^{(1)}$ and $\rho^{(2)}$ is given by $\sqrt{2-2\sqrt{\mathcal{F}(\rho^{(1)},\rho^{(2)})}}$, with the fidelity $\mathcal{F}(\rho^{(1)},\rho^{(2)})=\operatorname{Tr}(\sqrt{\sqrt{\rho^{(1)}}\rho^{(2)}\sqrt{\rho^{(1)}}}).$ 
\item The Bures distance between $\rho$ and a pure state $\sigma$ is reduced to $d_B\left(\rho, \sigma\right)=\sqrt{2-2 \sqrt{\operatorname{Tr}\left(\rho \sigma\right)}}.$

\item The Bures distance between two elements in $\mathcal{S}_N \times \mathcal{S}_N$ can be defined as $${d}_B\left((\rho^{(1)}, \tilde{\rho}^{(1)}),(\rho^{(2)}, \tilde{\rho}^{(2)})\right)=d_B\left(\rho^{(1)}, \rho^{(2)}\right)+d_B\left(\tilde{\rho}^{(1)}, \tilde{\rho}^{(2)}\right).$$ 
\end{itemize}

\end{definition}
{Denote by  $\mathbb P^{\prime}$ the probability measure, equivalent to $\mathbb P,$ that makes the observation process $Y_t$ a Wiener process. The existence of such a probability measure is guarenteed by Girsanov theorem. Its corresponding expectation is defined by $\mathbb E^{\prime.}$}
\begin{theorem}
 For system~\eqref{coupled2}, with no external control ($u \equiv 0$) and initial condition $\bar\rho_0 \in \mathcal{S}_N$, the set of invariant quantum states $\bar{E}:=\left\{{A}_{0}, \ldots, {A}_{2J}\right\}$ is exponentially stable both in mean and almost surely with respect to the probability measure $\mathbb P'$.
The average and sample Lyapunov exponent are less than or equal to $- \eta/2$. 

\end{theorem}
\begin{proof}
The proof is similar to \cite{liang2019exponential}. For completeness, we present the sketch of the proof. We consider
the candidate Lyapunov function
$$
V({\rho}_{\theta})=\frac{1}{2} \sum_{\substack{\bar n, \bar m=0 \\ \bar n \neq \bar m}}^{2J} \sqrt{\operatorname{Tr}\left({\rho}_{\theta}{A}_{\bar{n}}\right) \operatorname{Tr}\left({\rho}_{\theta}{A}_{\bar{m}}\right)},
$$
where $V({\rho}_{\theta})=0$ if and only if ${\rho}_{\theta} \in \bar{E}$.
First, 
one can show that $\mathcal L V({\rho}_{\theta})\leq -\frac{\eta}{2}V({\rho}_{\theta}).$ Then by using It\^{o}'s formula and the Grönwall inequality, 
one obtains an upper bound on the expected value of the Lyapunov function as a function of time, i.e, $\mathbb{E^\prime}\left(V\left({\rho}_{\theta_t}\right)\right)  \leq V\left(\bar\rho_0\right) e^{-\frac{ \eta }{2} t}.$
Moreover, by comparing $V$ with the Bures distance, one gets the following inequality:
$$
\mathbb{E^\prime}\left(d_B\left({\rho}_{\theta_t}, \bar{E}\right)\right) \leq \frac{C_2}{C_1} d_B\left(\bar\rho_0, \bar{E}\right) e^{-\frac{ \eta }{2} t}, 
$$
with  $C_1=1/2$ and $C_2=J(2J+1)$. This means that the set $\bar{E}$ is exponentially stable in  mean with  average Lyapunov exponent less than or equal to $- \eta/2$. This implies the {almost sure convergence}, see \cite{liang2019exponential} for more details.

\end{proof}
The above theorem does not guarantee that the trajectories of the equations~\eqref{coupled1} and~\eqref{coupled2} select the same limit. The approaches developed in \cite{benoist2014large} and \cite{bompais2023asymptotic} suggest that these limits coincide (as we assumed that $\mathrm{Tr}(\bar \rho_0 A_k)\neq 0$ for every $k$). Consequently, it is reasonable to 
guess that a feedback control mechanism, utilizing the quantum projection filter, can be employed to stabilize the system and steer it toward a desired eigenstate of $L$, similar to the approaches taken in previous studies like \cite{liang2021robust}, \cite{liang2019exponential}, and \cite{Vanhandel3}. In the following, we give sufficient conditions on the feedback controller ensuring  stabilization.

Concretely, we will give conditions that ensure the exponential stabilization of the pair $\left(\rho_t, {\rho}_{\theta_t}\right)$ toward the target state $\left({A}_{\bar{n}}, {A}_{\bar{n}}\right)$, with $\bar{n}\in\{0,\ldots,2J\}.$ 
For this purpose, we consider the following assumptions.






\medskip

{\textbf{A1}:} 
$u \in \mathcal{C}^1\left(\mathrm{clos}(\mathcal{M}), \mathbb{R}\right),$
$u\left({A}_{\bar{n}}\right)=0$ and $u(\rho) \neq 0$ for all $\rho \in\left\{{A}_0, \ldots, {A}_{2J}\right\} \setminus {A}_{\bar{n}}$.

\medskip

The above hypothesis ensures that the coupled system~\eqref{coupled1}-\eqref{coupled2} contains exactly $N$ equilibria, given by $\left( A_{{n}}, A_{\bar{n}}\right)$, with $n\in\{0,\ldots,2J\}.$

\medskip

{\textbf{A2}:} $|u(\rho)| \leq c\left(1-\mathrm{Tr}(\rho A_{\bar{n}})\right)^p$ for 
$\rho\in\mathcal{M}$
with $p>1 / 2$ and for some constant $c>0$.

\medskip

In the case $\bar{n} \in \{0,2J\}$ the assumption above  gives instability of $\left(A_n, A_{\bar{n}}\right)$ for $n \neq \bar{n}$.

\medskip

{\textbf{A3}:} $u(\rho)=0$ for all $\rho \in B_{\varepsilon}\left(A_{\bar{n}}\right)\cap \mathcal{M}$, for a sufficiently small $\varepsilon>0$.

\medskip

The above hypothesis is needed in the case $\bar{n} \notin\{0,2J\}$ for proving the instability of the target state $\left(A_n, A_{\bar{n}}\right)$ 
as well as  the reachability of an arbitrary neighborhood of $\left(A_{\bar{n}}, A_{\bar{n}}\right)$. 
Informally speaking, reachability means that any neighborhood of the target state can be reached in finite time. 

\medskip

Our final assumption yields, in combination with the previous assumptions, the reachability of the target state for the coupled system.

Let us define $P_n(\rho):=J-n-\operatorname{Tr}(J_z\rho)$ and $\mathbf{P}_n:=\left\{\rho\in\mathcal{M}\mid P_n(\rho)=0\right\}$
 for $n \in\{0, \ldots, 2 J\},$ and  $\mathscr{V}_z(\rho):=\operatorname{Tr}\left(J_z^2 \rho\right)-\operatorname{Tr}\left(J_z \rho\right)^2$. Set also $\Theta_n(\rho):=\operatorname{Tr}(i[J_y,\rho]A_n).$

\medskip

{\textbf{A4}:} For all $\rho \in {\mathbf{P}}_{\bar{n}} \setminus {A}_{\bar{n}},$ 
 \begin{equation*}
 2 \eta  \mathscr{V}_z(\rho) \operatorname{Tr}(\rho A_{\bar{n}}) >u(\rho) \Theta_{\bar{n}}(\rho). 
 \label{Variance function}
\end{equation*}
In the following, we establish a result concerning the exponential stability of the target state $\left(A_{\bar{n}}, A_{\bar{n}}\right)$ for the coupled equations~\eqref{coupled1}-\eqref{coupled2}.
\begin{theorem} 
Assume either that $\bar n\in\{0,2J\}$ and \textbf{A1}, \textbf{A2} hold true, or that $\bar n\notin\{0,2J\}$ and \textbf{A1}, \textbf{A3} and \textbf{A4} are verified. Assume moreover that  there exists a function $V(\rho, {\rho}_{\theta}),$ which is continuous on $\mathcal{S}_N \times \mathrm{clos}(\mathcal{M})$, and twice continuously differentiable on an almost surely invariant subset $\Gamma$ of $\mathcal{S}_N \times \mathrm{clos}(\mathcal{M})$, which includes $\operatorname{int}\left(\mathcal{S}_N\right) \times \mathcal{M},$ satisfying the following conditions 
\begin{itemize}
\item[(i)] there exist positive constants $C_1$ and $C_2$ such that $$C_1 {d}_B\left((\rho, {\rho}_{\theta}),\left(A_{\bar{n}}, A_{\bar{n}}\right)\right) \leq V(\rho, {\rho}_{\theta}) \leq C_2 {d}_B\left((\rho, {\rho}_{\theta}),\left(A_{\bar{n}}, A_{\bar{n}}\right)\right)$$ for all 
$(\rho, {\rho}_{\theta}) \in \mathcal{S}_N \times \mathrm{clos}(\mathcal{M})$, 
\item[(ii)] there exists $C>0 $ such that $$\limsup _{(\rho, {\rho}_{\theta}) \rightarrow\left(A_{\bar{n}}, A_{\bar{n}}\right)} \frac{\mathscr{L} V(\rho, {\rho}_{\theta})}{V(\rho, {\rho}_{\theta})} \leq-C.$$
\end{itemize}
Then, starting from $\Gamma$,  the target state $\left(A_{\bar{n}}, A_{\bar{n}}\right)$ is almost surely exponentially stable for the coupled system~\eqref{coupled1}-\eqref{coupled2} with a sample Lyapunov exponent less than or equal to $-C-\frac{K}{2}$, where $K:=$ $\liminf _{(\rho, {\rho}_{\theta}) \rightarrow\left(A_{\bar{n}}, A_{\bar{n}}\right)} g^2(\rho, {\rho}_{\theta})$ and $g(\rho, {\rho}_{\theta}):=\frac{\partial V(\rho, {\rho}_{\theta})}{\partial \rho} \frac{\mathcal G(\rho)}{V(\rho, {\rho}_{\theta})}+\frac{\partial V(\rho, {\rho}_{\theta})}{\partial {\rho}_\theta} \frac{\mathcal G({\rho}_{\theta})}{V(\rho, {\rho}_{\theta})}$. \label{General exp}
\end{theorem}
\paragraph{\it{Sketch of the proof}} The proof is just an adaptation of the one given in \cite[Theorem 4.11]{liang2021robust} for the coupled equations of quantum filters and its corresponding projection filter. Roughly speaking, the instability of the equilibria $(A_n,A_{\bar n})$ for $n\neq\bar n$ can be shown by applying assumptions  \textbf {A1} and \textbf{A2} for the case $\bar n\in\{0,2J\}$, and assumptions  \textbf {A1} and \textbf{A3} for $\bar n\notin\{0,2J\}.$ The reachability of the target state $(A_{\bar n},A_{\bar n})$ can be shown by assuming the above hypotheses and applying similar arguments as in \cite{liang2021robust}. The key point for the validity of instability of the equilibria $(A_n,A_{\bar n})$ for $n\neq\bar n$ and the reachability of the target state $(A_{\bar n},A_{\bar n})$ is based on the fact that the diagonal element  of the projection filter~\eqref{coupled2}, i.e., $\operatorname{Tr}({\rho}_{\theta_t}A_{ n})$ and the diagonal element of the estimate filter~\eqref{estimate}, i.e.,  $\operatorname{Tr}(\hat{\rho}_tA_{ n})$ follow the same dynamical evolution given by 
\begin{equation}
    \begin{aligned}
d \operatorname{Tr}( {\rho}_{\theta_t}A_{ n})&=-u \Theta_{n}({\rho}_{\theta_t})dt+4\eta P_{\bar n}({\rho}_{\theta_t})\operatorname{Tr}({\rho}_{\theta_t}A_n)\mathcal{T}(\rho, {\rho}_{\theta_t})dt&\\&+4\sqrt{\eta} P_{n}({\rho}_{\theta_t})\operatorname{Tr}({\rho}_{\theta_t}A_n) d W_t,
\label{Diagonall}
\end{aligned}
\end{equation}
where $\mathcal{T}(\rho, {\rho}_{\theta}):= \operatorname{Tr}\left(J_z \rho\right)- \operatorname{Tr}\left(J_z {\rho}_{\theta}\right).$
Finally, similar to the proof of \cite[Theorem 4.11]{liang2021robust}, one can show by using conditions (i) and (ii) the almost sure convergence toward the target state $(A_{\bar n},A_{\bar n})$ and conclude that 
\[{\limsup}_{t\rightarrow\infty}\frac 1 t\log V(\rho_t,\rho_{\theta_t})\leq -C-\frac K 2,\quad \mbox{a.s.}\]

\medskip

In the following, we give the explicit rate of convergence for the coupled system~\eqref{coupled1}-\eqref{coupled2}.
\begin{theorem}
Consider the coupled system~\eqref{coupled1}-\eqref{coupled2}
and suppose that $\rho_0\in\operatorname{int}\left(\mathcal{S}_N\right).$ 
\begin{itemize}
\item [-] Assume that hypotheses \textbf{A1} and \textbf{A2} hold. Then 
the pair $\left({A}_{\bar{n}}, {A}_{\bar{n}}\right)$ for $\bar n\in\{0,2J\}$ is almost surely exponentially stable with the sample Lyapunov exponent less than or equal to $-\eta.$ 

\item[-] Assume that hypotheses  \textbf{A1}, \textbf{A3}, and \textbf{A4} hold. Then the target state $\left({A}_{\bar{n}}, {A}_{\bar{n}}\right)$ with $\bar n\notin\{0,2J\}$ is almost surely exponentially stable with the sample Lyapunov exponent less than or equal to $-\frac{\eta}{2}.$ 
\end{itemize}
\label{exponential stability}

\end{theorem}
\paragraph{\it{Sketch of the proof}} The proof is based on considering the Lyapunov candidates 
\[
V(\rho,\rho_\theta) = \left\{\begin{array}{ll}\sqrt{1-\operatorname{Tr}(\rho A_{\bar n})}+\sqrt{1-\operatorname{Tr}(\rho_{\theta} A_{\bar n})},&\mbox{ for }\bar n\in\{0,2J\},\\
\sum_{n \neq \bar{n}} \sqrt{\operatorname{Tr}(\rho A_n)}+\sum_{n \neq \bar{n}} \sqrt{\operatorname{Tr}(\rho_{\theta} A_n)},& \mbox{ for }\bar n\notin\{0,2J\},
\end{array}\right.
\]
and applying the previous theorem. This is similar to the proofs of \cite[Theorems 4.15 and 4.17]{liang2021robust} as the latter Lyapunov functions only depend on the diagonal elements of $\rho$ and $\rho_\theta.$ 
\section{Feedback design based on the exponential family structure}\label{D}
By the previous sections, we know that the projection filter dynamics can be equivalently studied through the equation~\eqref{Quantum projection}. In this section, based on the general results established in the previous section for $N$-level quantum systems, we design a stabilizing feedback based on a new  parametrization of the exponential family. This makes the previous study more interesting since 
working directly with the new parametrization reduces the complexity of the considered filter.

{Based on the fact that $\check\rho_{\theta}$ can be written as 
\begin{equation*}
\check\rho_{\theta} = \sum_{k,j=0}^{2J}e^{\frac{\theta_k+\theta_j}{2}}A_k\bar{\rho}_0A_j
\label{reduced-rho},
\end{equation*}} we can rewrite the normalized projection filter ${\rho}_{\theta}$ in terms of $\theta$ as follows \begin{equation}
{\rho}_{\theta}=\frac{\check{\rho}_{\theta}}{\operatorname{Tr}\left(\check{\rho}_{\theta}\right)}=\frac{\sum_{k,j=0}^{2J}e^{\frac{\theta_k+\theta_j}{2}}A_k\bar\rho_0 A_j}{\sum_{k=0}^{2J}e^{\theta_k}\operatorname{Tr}( A_k\bar\rho_0)}. \label{Manifold in theta}
\end{equation}

Note that $\rho_\theta \neq A_{k}$ for every  $\theta\in\mathbb R^{2J+1}$ and $k=0,\ldots,2J$. In particular the target state $A_{\bar n}$ cannot be expressed as an equilibrium of the system in the coordinates $\theta$. 
Furthermore, the parametrization $\theta\mapsto\rho_\theta$ is redundant in the sense that, for every shift $\theta'=\theta+\mu(1,\ldots,1)^T$ with $\mu\in \mathbb{R}$, one has $\rho_{\theta'}=\rho_\theta$.
In order to overcome 
these issues, we define a new  coordinate system: let us take $\xi=\left(\xi_0,\ldots,\xi_{\bar n-1},\xi_{\bar n+1},\ldots,\xi_{2J} \right)^{T}\in\mathbb{R}^{2J}$, where $\xi_k=e^\frac{\theta_k-\theta_{\bar n}}{2}$, $k \neq \bar n $.  Then~\eqref{Manifold in theta} can be rewritten in the following form 
\begin{equation}
{\rho}_{\xi}=\frac{A_{\bar n}\bar\rho_0 A_{\bar n}+\sum_{k \neq \bar n}\xi_k(A_k\bar\rho_0 A_{\bar n}+A_{\bar n}\bar\rho_0A_k)+\sum_{k,j \neq \bar n}\xi_k\xi_jA_k\bar\rho_0 A_j}{\operatorname{Tr}(\bar{\rho}_0A_{\bar n})+\sum_{k \neq \bar n}\xi_k^2\operatorname{Tr}( A_k\bar\rho_0)}.  \label{Manifold in xi}
\end{equation}


{In this representation, the target state $A_{\bar n}$ corresponds to the value $\xi=0$. The other eigenstates $A_k, k \neq \bar n$ of the measurement operator can be obtained as the limits of $\rho_{\xi}$ as $\xi$ tends to infinity and $\frac{|\xi_k|}{\|\xi\|}$ tends to one.}

The dynamics of the projection filter can now be described in terms of  $\xi$. 
To obtain the dynamics of $\xi_k$, it is sufficient to apply It\^o's formula to get\footnote{For the sake of simplicity, from now on we use the convention $\xi_{\bar n}=1.$}
\begin{align}
d\xi_k(t)&=- u(\xi(t))\sum_{p\neq k}  \frac{\operatorname{Tr}\left(iJ_y (A_p\bar\rho_0A_k-A_k\bar\rho_0A_p)\right)}{2\operatorname{Tr}(\bar\rho_0A_k)}\xi_p(t) dt \nonumber\\ &+{u(\xi(t))}\sum_{p\neq \bar n} \frac{\operatorname{Tr}\left(iJ_y  (A_p\bar\rho_0A_{\bar n}-A_{\bar n}\bar\rho_0A_p)\right)}{2\operatorname{Tr}(\bar\rho_0A_{\bar n})}\xi_p(t)\xi_k(t) dt \nonumber\\ &
+\eta \left(\frac{3\lambda_{\bar n}^{2}}{2}-\frac{\lambda_{k}^{2}}{2} -\lambda_k \lambda_{\bar n}\right)\xi_k(t) dt +2\eta \left(\lambda_k-\lambda_{\bar n}\right)\operatorname{Tr}(J_z \rho_t)\xi_k(t)dt \nonumber\\ & +\sqrt{\eta}\left(\lambda_k-\lambda_{\bar n}\right)\xi_k(t) dW_t. \label{Simplified filter}
\end{align}

To address the stabilization problem, we shall from now on work with the coupled system \eqref{coupled1}-\eqref{Simplified filter}, describing the evolution of $(\rho,\xi).$ The 
assumptions on the feedback control $u(\xi)$ given in the previous section can be adapted 
in the current framework
as follows.

\medskip 

{\textbf{A1}$'$:} $u \in \mathcal{C}^1\left(\mathbb{R}^{2J}, \mathbb{R}\right),  u\left(0\right)=0$ and  {
$\displaystyle \liminf_{\xi \rightarrow\infty ,\, \frac{|\xi_k|}{\|\xi\|} \rightarrow 1}  |u(\xi)| > 0$ for all $k \in \{0,\ldots,2J\}.$}
{
 This ensures that $(A_k,0) \in\mathcal S_N\times\mathbb R^{2J}$ for $k \in \{0, \ldots, 2J \}$, correspond to the equilibria for the coupled system~\eqref{coupled1}-\eqref{Simplified filter}, and that the pair $(\rho,\rho_\xi)$ does not converge to $(A_k,A_h)$ for $h \neq \bar n$.}\\
 
\medskip

{\textbf{A2}$'$:} $|u(\xi)| \leq c\max\{\|\xi\|^q,1\}$ for $\xi \in \mathbb{R}^{2J}$ with $q>1$ and for some constant $c>0$.

\medskip

{\textbf{A3}$'$:} $u(\xi)=0$ if $\|\xi\| \leq \epsilon$, for a sufficiently small $\epsilon>0$.

\medskip

{{\textbf{A4}$'$:} For every $\xi\neq 0$ such that $\sum_{k\neq \bar n}(\lambda_{\bar n}-\lambda_k)\xi_k^{2}\operatorname{Tr}( A_k\bar\rho_0) = 0$ it holds
 \begin{equation*}
 2 \eta  \mathscr{V}_z(\rho_\xi) \operatorname{Tr}(\rho_\xi  A_{\bar{n}}) >u(\xi) \Theta_{\bar{n}}(\rho_\xi). 
 \label{Variance function 2}
\end{equation*}}

\medskip


The following theorem gives an analogue of Theorem~\ref{exponential stability} in the present framework. In order to apply the stability notions in Definition~\ref{stability}, we endow the space $\mathcal{S}_N\times \mathbb{R}^{2J}$ with the distance $d\left((\rho_1,\xi_1),(\rho_2,\xi_2)\right) = d_B(\rho_1,\rho_2)+\|\xi_1-\xi_2\|.$
\begin{theorem}
    Consider the system~\eqref{coupled1}-\eqref{Simplified filter} with
    $\rho_0 \in \operatorname{int}\left(\mathcal{S}_N\right)$ and $\xi_0=\left(1,\ldots,1\right)^T$, {and assume that the solution of the equation~\eqref{Simplified filter} is well-defined on $[0,+\infty)$, that is, that the solution does not blow up in finite time almost surely.}
\begin{itemize}
\item [-] Assume that $\bar n\in\{0,2J\}$ and hypotheses \textbf{A1}$'$ and \textbf{A2}$'$ hold. Then 
the pair $\left({A}_{\bar{n}}, 0\right)\in\mathcal S_N\times\mathbb R^{2J}$  is almost surely exponentially stable with the sample Lyapunov exponent less than or equal to $-\eta.$ 

\item[-] Assume that $\bar n\notin\{0,2J\}$ and hypotheses \textbf{A1}$'$, \textbf{A3}$'$, and \textbf{A4}$'$ hold. Then the target state $\left({A}_{\bar{n}}, 0\right)\in\mathcal S_N\times\mathbb R^{2J}$   is almost surely exponentially stable with the sample Lyapunov exponent less than or equal to $-\frac{\eta}{2}.$ 
\end{itemize}
\label{ simplified exp stabilit}
\end{theorem}
\begin{proof} 
It is easy to verify that the conditions \textbf{A1}$'$, \textbf{A3}$'$ and \textbf{A4}$'$ imply \textbf{A1}, \textbf{A3} and \textbf{A4}, respectively. 
Furthermore, we observe that
\[1-\mathrm{Tr}(\rho_\xi A_{\bar n})=\frac{\sum_{k\neq\bar n} \xi_k^2 \mathrm{Tr}(A_k\bar\rho_0)}{\mathrm{Tr}(A_{\bar n}\bar\rho_0)+\sum_{k\neq\bar n} \xi_k^2 \mathrm{Tr}(A_k\bar\rho_0)}\geq \|\xi\|^2 \frac{\min_{k\neq \bar n}\mathrm{Tr}(A_k\bar\rho_0)}{2\mathrm{Tr}(A_{\bar n}\bar\rho_0)}\]
for $\xi$ small enough. Hence \textbf{A2}$'$ implies \textbf{A2} for $\rho=\rho_\xi$ with $\xi$ small enough, that is, \textbf{A2} is verified in a neighborhood of $A_{\bar n}$ in $\mathcal M$. By using $|u(\xi)| \leq c$ and the boundedness from below of $1-\mathrm{Tr}(\rho A_{\bar{n}})$ outside of that neighborhood, we can conclude the validity of \textbf{A2} for $\rho\in\mathcal M.$ Then by Theorem~\ref{exponential stability} the system is almost surely exponentially stable. Concerning sample Lyapunov exponents, it is enough to observe that in a small enough neighborhood of the target state the following holds
\[\frac12 \min_{k\neq \bar n}\sqrt{\frac{\mathrm{Tr}(A_{k}\bar\rho_0)}{\mathrm{Tr}( A_{\bar n}\bar\rho_0)}}\|\xi\|\leq d_B(\rho_\xi, A_{\bar n}) \leq 2\max_{k\neq \bar n}\sqrt{\frac{\mathrm{Tr}(A_{k}\bar\rho_0)}{\mathrm{Tr}(A_{\bar n}\bar\rho_0)}}\|\xi\|,\]
which implies that 
\[\limsup_{t\to\infty}\frac1{t}\log\|\xi\| = \limsup_{t\to\infty}\frac1{t}\log d_B(\rho_{\xi(t)},A_{\bar n})\]
for every trajectory converging to the equilibrium.
\end{proof}

{Here we give concrete examples of feedbacks which satisfy the previous assumptions and hence stabilize the coupled system~\eqref{coupled1}-\eqref{Simplified filter}  toward the specific target state $\left(A_{\bar{n}}, 0\right)$.  
\begin{Corollary}
Consider the coupled system~\eqref{coupled1}-\eqref{Simplified filter} where the initial condition $\rho_0$ belongs to $\operatorname{int}\left(\mathcal{S}_N\right)$ and $\xi_0=\left(1,\ldots,1\right)^T$. Assume  $\operatorname{Tr}\left(J_y  (A_p\bar\rho_0A_{\bar n}-A_{\bar n}\bar\rho_0A_p)\right)=0$ for $p \neq \bar n$.
\begin{itemize}
\item[-] If $\bar{n} \in\{0, 2J\},$ then the feedback controller 
\begin{equation}
    u_{\bar{n}}(\xi)=\alpha\frac{\|\xi\|^\beta}{1+\|\xi\|^\beta}
    \label{control -1}
\end{equation}
with $\alpha>0$ and $\beta \geq 1$ stabilizes the coupled system toward the state $\left(A_{\bar{n}}, 0\right)$
with a sample Lyapunov exponent no greater than $-\eta$. 
\item[-] If $\bar{n} \in\{1,\ldots, 2J-1\},$ then any feedback controller of the form
\begin{equation}
    u_{\bar{n}}(\xi)=\frac{1}{1+\|\xi\|^{2\beta}}{\Big(\sum_{k\neq \bar n}(\lambda_{\bar n}-\lambda_k)\xi_k^{2}\operatorname{Tr}( A_k\bar\rho_0)\Big)^\beta f\left(\xi\right)}, \label{Control -2}
\end{equation}
with  $\beta \geq 1$ and $f\in\mathcal{C}^1(\mathbb{R}^{2J},\mathbb{R})$ satisfying $f|_{B_\varepsilon (0)}\equiv 0,$ $\|f\|<\infty$, and $ \liminf_{\xi \rightarrow\infty ,\, \frac{|\xi_k|}{\|\xi\|} \rightarrow 1} |f(\xi)| \neq 0,$ stabilizes the coupled system toward the state $\left(A_{\bar{n}}, 0\right).$
The sample Lyapunov exponent is no greater than 
$-\eta/2$.
\end{itemize}
\label{feedback -1}
\end{Corollary} 
Note that under the condition $\operatorname{Tr}\left(J_y  (A_p\bar\rho_0A_{\bar n}-A_{\bar n}\bar\rho_0A_p)\right)=0$ for $p \neq \bar n$, the solution of the equation~\eqref{Simplified filter} is well-defined on $[0,\infty)$ since the equation does not contain quadratic terms. Also, the feedback given in~\eqref{Control -2} vanishes on the domain of $\xi$ considered in  the assumption \textbf{A4}$'$, which is then satisfied. The other required assumptions given in Theorem \ref{ simplified exp stabilit} can be easily verified.
An example of function $f$ satisfying the conditions of the previous corollary is given by 
\begin{equation}
f(\xi)=\max\left(0,\frac{(\|\xi\|-c)^3}{1+\|\xi\|^3}\right),
\label{f}
\end{equation}
for $c>0.$
}

{In the general case where $\operatorname{Tr}\left(J_y  (A_p\bar\rho_0A_{\bar n}-A_{\bar n}\bar\rho_0A_p)\right) \neq 0$ for some $p \neq \bar n$, ensuring that the solution  of  the dynamics~\eqref{Simplified filter} is well-defined in $[0,\infty)$ is not trivial, since the equation contains a quadratic term in  $\xi$. 
In the following lemma, we give a condition on the feedback controller which ensures this property. Let us denote by $\tau_e$ the explosion time.
\begin{lemma}
    Consider the coupled system~\eqref{coupled1}-\eqref{Simplified filter} and set \begin{equation}
        u(\xi)=-\phi(\xi)\sum_{p\neq \bar n} \frac{\operatorname{Tr}\left(iJ_y  (A_p\bar\rho_0A_{\bar n}-A_{\bar n}\bar\rho_0A_p)\right)}{2\operatorname{Tr}(\bar\rho_0A_{\bar n})}\xi_p , \label{feedback law}
        \end{equation} where $\phi(\cdot)$ is a real bounded non-negative function such that $|\phi(\xi)| \leq \frac{\Delta}{\|\xi\|+1}$ with $\Delta>0$. Then every solution of the equation~\eqref{Simplified filter} is well-defined on $[0,+\infty),$ i.e., $\mathbb P(\tau_e<\infty)=0.$
        \label{explosion}
\end{lemma}
\begin{proof}
We consider the dynamics of $\|\xi\|^2=\sum_{k\neq \bar n} \xi_k^2$. We have
\begin{align}
d\xi_k^2(t)&=-2\xi_k(t) u(\xi(t))\sum_{p\neq k}  \frac{\operatorname{Tr}\left(iJ_y (A_p\bar\rho_0A_k-A_k\bar\rho_0A_p)\right)}{2\operatorname{Tr}(\bar\rho_0A_k)}\xi_p(t) dt \nonumber\\ &+2\xi_k^2(t){u(\xi(t))}\sum_{p\neq \bar n} \frac{\operatorname{Tr}\left(iJ_y  (A_p\bar\rho_0A_{\bar n}-A_{\bar n}\bar\rho_0A_p)\right)}{2\operatorname{Tr}(\bar\rho_0A_{\bar n})}\xi_p(t) dt \nonumber\\ &
+2\eta \left(\frac{3\lambda_{\bar n}^{2}}{2}-\frac{\lambda_{k}^{2}}{2} -\lambda_k \lambda_{\bar n}\right)\xi_k^2(t) dt+4\eta \left(\lambda_k-\lambda_{\bar n}\right)\operatorname{Tr}(J_z \rho_t)\xi_k^2(t) dt \nonumber\\ & + \eta  \left(\lambda_k-\lambda_{\bar n}\right)^2 \xi_k^2(t) dt +2\sqrt{\eta}\left(\lambda_k-\lambda_{\bar n}\right)\xi_k^2(t) dW_t. \label{Simplified filter correction}
\end{align}
Since the feedback~\eqref{feedback law}  is bounded, all the terms in the right hand side of~\eqref{Simplified filter correction} are sublinear in $\|\xi\|^2,$ except for the second term 
which  is negative. As a consequence, the infinitesimal generator of $\|\xi\|^2$ satisfies $\mathscr{L}\left(\|\xi\|^2 \right)\leq c \|\xi\|^2$ with $c>0.$

Now, define the stopping time $\tau_\varepsilon:=\inf \left\{t \geq 0 \mid \|\xi(t)\|^2 \in [\frac{1}{\varepsilon}, \infty) \right\}$. 
By It\^o's formula, for every $T>0$ we get 
\begin{equation}
\mathbb E(e^{-c (T \wedge \tau_\varepsilon)}\|\xi(T \wedge \tau_\varepsilon)\|^2)= \|\xi(0)\| ^2 + \mathbb E \left(\int_0^{T\wedge \tau_{\varepsilon} } \mathscr{L} (e^{-cs}\|\xi(s)\|^2) d s\right) \leq \|\xi(0)\| ^2,
\label{ito}
\end{equation} 
the last inequality coming from the fact that $\mathscr{L}\left(e^{-ct}\|\xi\|^2 \right)\leq 0.$

Now we note that by the definitions of $\tau_e$ and $\tau_\varepsilon,$ we have $\tau_\varepsilon\leq\tau_e.$ We have the following bound

\begin{align*}
\mathbb E\left(e^{-cT} \varepsilon^{-1} \mathbf{1}_{\{\tau_e \leq T \}}\right) &\leq  \mathbb E\left(e^{-c\tau_\varepsilon} \|\xi(\tau_\varepsilon)\| ^2 \mathbf{1}_{\{\tau_e \leq T \}}\right)\\
&=\mathbb E\left(e^{-c(T\wedge\tau_\varepsilon)} \|\xi(T\wedge\tau_\varepsilon)\| ^2 \mathbf{1}_{\{\tau_e \leq T \}}\right)\\
&\leq\|\xi(0)\| ^2,
\end{align*}
where for the last inequality we used the inequality \eqref{ito}.

Thus, $\mathbb{E}\left( \mathbf{1}_{\{\tau_e \leq T\}}\right)=\mathbb P(\tau_e \leq T)\leq \varepsilon e^{cT} \|\xi(0)\| ^2.$
Now it is sufficient to let $\varepsilon$ tends to zero, to conclude that 
$\mathbb P(\tau_e \leq T)=0.$ The proof is complete since $T$ is chosen arbitrarily. 
\end{proof} 
Based on the results discussed above, below we provide further examples of feedbacks  stabilizing the coupled system~\eqref{coupled1}-\eqref{Simplified filter}  toward the specific target state $\left(A_{\bar{n}}, 0\right)$. 
\begin{Corollary}
Consider the coupled system~\eqref{coupled1}-\eqref{Simplified filter} where the initial condition $\rho_0$ belongs to $\operatorname{int}\left(\mathcal{S}_N\right)$ and $\xi_0=\left(1,\ldots,1\right)^T$. Assume $\operatorname{Tr}\left(J_y  (A_p\bar\rho_0A_{\bar n}-A_{\bar n}\bar\rho_0A_p)\right) \neq 0$ for every $p \neq \bar n$.
\begin{itemize}
\item[-] If $\bar{n} \in\{0, 2J\},$ then the feedback controller 
\begin{equation}
    u_{\bar{n}}(\xi)=-\alpha\frac{\|\xi\|^\beta}{1+\|\xi\|^{\beta +1}}\sum_{p\neq \bar n} \frac{\operatorname{Tr}\left(iJ_y  (A_p\bar\rho_0A_{\bar n}-A_{\bar n}\bar\rho_0A_p)\right)}{2\operatorname{Tr}(\bar\rho_0A_{\bar n})}\,\xi_p
    \label{control 1}
\end{equation}
with $\alpha>0$ and $\beta \geq 1$ stabilizes the coupled system toward the state $\left(A_{\bar{n}}, 0\right)$
with a sample Lyapunov exponent no greater than $-\eta$. 
\item[-] If $\bar{n} \in\{1,\ldots, 2J-1\},$ then any feedback controller of the form
\begin{equation}
     u_{\bar{n}}(\xi)=-\frac{g(\xi)}{1+\|\xi\|^{2\beta + 1}} \sum_{p\neq \bar n} \frac{\operatorname{Tr}\left(iJ_y  (A_p\bar\rho_0A_{\bar n}-A_{\bar n}\bar\rho_0A_p)\right)}{2\operatorname{Tr}(\bar\rho_0A_{\bar n})}\,\xi_p,
     \label{Control 2}
\end{equation}
where  $$g(\xi)=\Big(\sum_{k\neq \bar n}(\lambda_{\bar n}-\lambda_k)\xi_k^{2}\operatorname{Tr}( A_k\bar\rho_0)\Big)^{\beta}\!f(\xi)$$ with $\beta \geq 1$ and $f\in\mathcal{C}^1(\mathbb{R}^{2J},\mathbb{R})$ satisfying $f|_{B_\varepsilon (0)}\equiv 0$, $\|f\|<\infty$, and { $ \liminf_{\xi \rightarrow\infty ,\, \frac{|\xi_k|}{\|\xi\|} \rightarrow 1} |f(\xi)| \neq 0$,}  stabilizes the coupled system toward the state $\left(A_{\bar{n}}, 0\right).$ The sample Lyapunov exponent is no greater than
$-\eta/2$.
\label{feedback 1}
\end{itemize}
\end{Corollary} 
Here the feedbacks are in the form given in Lemma \ref{explosion} and satisfy the other required assumptions given in Theorem \ref{ simplified exp stabilit}.
Similar as before, the function $f$ can be chosen as in~\eqref{f}.}
\section{Simulations}\label{E}
In this section, we test our previous results through numerical simulations of a four-level quantum angular momentum system. Here $N=4$ and $J=\frac{3}{2}.$  A feedback control satisfying the conditions of Theorem \ref{exponential stability} is given by  \begin{equation}
u({\rho}_\theta)=\alpha\left(1-\operatorname{Tr}\left({\rho}_\theta  A_{0}\right)\right)^\beta ,\label{control 3}
\end{equation}
where ${\rho}_\theta$ is the projection filter and ${A_0}$ is the target state.  By the spectral decomposition, the measurement operator can be written as $L=\sum_{j=0}^3\lambda_jA_j$, where $\lambda_0=\frac{3}{2}$, $ \lambda_1=\frac{1}{2}$, $ \lambda_2=-\frac{1}{2}$ and $ \lambda_3=-\frac{3}{2}$ are the distinct eigenvalues of $L$, corresponding to the four projection operators given by $A_0=\mathrm{diag}(1,0,0,0)$, $A_1=\mathrm{diag}(0,1,0,0)$, $A_2=\mathrm{diag}(0,0,1,0),$ and $A_3=\mathrm{diag}(0,0,0,1)$ respectively.

We set the detector efficiency $\eta=0.5$. 
Simulations have been done in the interval $[0,5]$, with step size $\delta=5\times 2^{-12}$. The initial states have been chosen as follows
\begin{align*}
\rho_0&=\begin{pmatrix}
    0.2 & 0 & 0 & 0 \\ 0 & 0.2 & 0 & 0 \\ 0 & 0 & 0.3 & 0 \\ 0 & 0 & 0 & 0.3
\end{pmatrix},
&
\bar{\rho}_0&=\begin{pmatrix}
0.2 & 0 & 0 & 0.1i \\ 0 & 0.2 & -0.1i & 0 \\ 0 & 0.1i & 0.3 & 0 \\ -0.1i & 0 & 0 & 0.3
\end{pmatrix}.
\end{align*} 
Figure \ref{fidelityevolution} shows the fidelity between the quantum filter and the quantum projection filter, the fidelity between the quantum filter and the target state $A_0,$ and 
 the fidelity between the quantum projection filter and the target state $A_0.$


 Regarding the method presented in Section \ref{D}, we apply the feedback laws proposed in Corollary \ref{feedback 1}. If we consider the filter equation in~\eqref{belavkin}, we should solve a stochastic differential equation in dimension $15.$ On the other hand, when using the presented projection method, we deal with a   stochastic differential equation in dimension $3.$ 
 Here we take $\rho_0,\bar\rho_0$ as before and we set $\xi_0=\left(1,1,1\right)^T.$

In Figures \ref{XiforA1} and \ref{XiforA2}, we show the convergence of the quantum filter $\rho$ to the target states $A_0$ and $A_1$ by applying the feedback laws \ref{control 1} and \ref{Control 2} respectively.
\begin{figure}
\centerline{\includegraphics[width=5in]{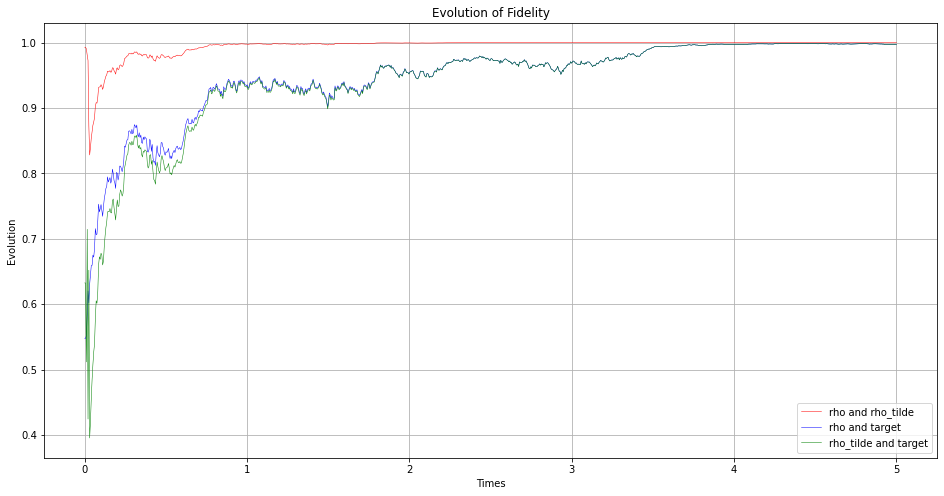}}
\caption{Stabilization of the coupled four-level quantum angular momentum system toward $(A_0,A_0)$ by the feedback control~\eqref{control 3} with $\alpha=10$ and $\beta=5$. }
\label{fidelityevolution}
\end{figure}
\begin{figure}
\centerline{\includegraphics[width=5in]{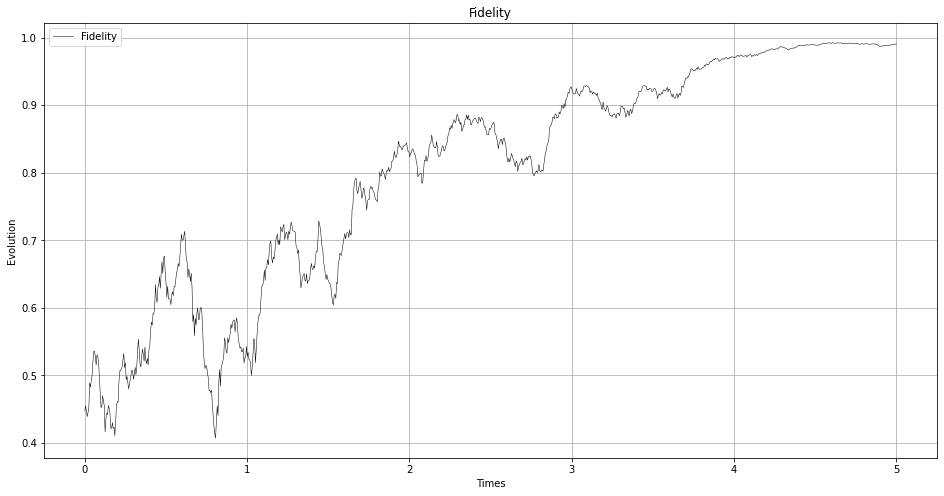}}
\caption{Convergence of the quantum filter~\eqref{coupled1} to the target state $A_0$, by applying the feedback control~\eqref{control 1} with $\alpha=2$ and $ \beta=2$.}
\label{XiforA1}
\end{figure}
\begin{figure}
\centerline{\includegraphics[width=5in]{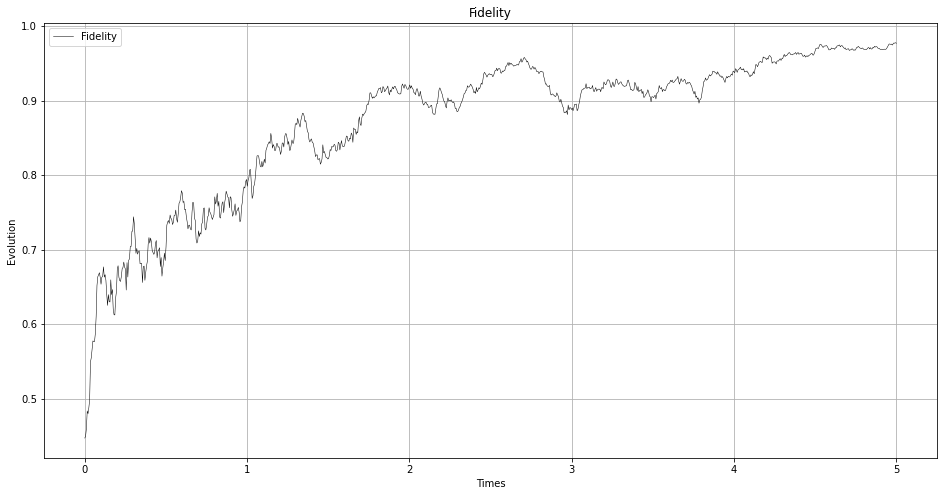}}
\caption{Convergence of the quantum filter~\eqref{coupled1} to the target state $A_1$, by applying the feedback control~\eqref{Control 2} with $\beta=2$ and $ f(\xi)=\max\left(0,\frac{(\|\xi\|-2)^3}{1+\|\xi\|^3}\right)$.}
\label{XiforA2}
\end{figure}

\section{Conclusion}\label{F}
In this paper we show that the projection filter can be applied in the feedback stabilization. To this aim, we demonstrate that the analysis in \cite{liang2021robust} can be adapted to our framework. Finally, we propose a new technique of feedback stabilization based on a new parameterization of the exponential family. This method reduces significantly the computational complexity of the real-time implementation of the feedback strategy.

Future research directions shall investigate alternative methods for simplifying the complexity of quantum filters, extend the scope of feedback stabilization based on approximate filters to broader contexts, and explore the application of projection filters in error correction.

\appendix
\section{Orthogonal projection}
\label{Orthogonal projection}
The set $\mathcal{M}$ defined in equation~\eqref{manifold} is, locally, a differential manifold {of dimension $m$} if the derivatives 
$$
\check{\partial}_k:=\frac{\partial \check{\rho}_\theta}{\partial \theta_k}
$$
form a linearly independent subset of the space $\mathcal{A}$  of Hermitian operators on $\mathbb C^N$. This can be easily verified if the operators $A_j$ form a family of mutually commuting projectors. We borrow from  the theory of quantum information geometry~\cite{Amari} some tools allowing to define a projection from  $\mathcal{A}$  onto the tangent space of $\check{\mathcal{M}}$ at the point  $\check{\rho}_\theta,$ denoted by $T_{\check{\rho}_\theta}\check{\mathcal{M}}$. Any vector in $T_{\check{\rho}_\theta}\check{\mathcal{M}}$ can be naturally identified with an element of $\mathcal{A}$, called the mixture representation or $m$-representation of the tangent vector. With this identification the elements $\check{\partial}_j$ form a basis of the tangent space $T_{\check{\rho}_\theta}\check{\mathcal{M}}$. The $m$-representation of a vector $X$ is indicated as $X^{(m)}$.

We will endow  $\check{\mathcal{M}}$ with a Riemannian metric. For this purpose we define the inner product on $\mathcal{A}$
$$
\langle\langle A, B\rangle\rangle_{\check{\rho}_\theta}=\frac{1}{2}{\rm Tr}(\check{\rho}_\theta AB + \check{\rho}_\theta BA),\qquad  \forall A,B \in \mathcal{A} .
$$
In light of the aforementioned inner product, we can define an additional representation of a tangent vector $X \in T_{\check{\rho}_\theta}\check{\mathcal{M}}$ called the $e$-representation. This representation corresponds to a Hermitian operator denoted as $X^{(e)}$ uniquely defined by the relation
$$
\langle\langle X^{(e)}, A\rangle\rangle_{\check{\rho}_\theta}=\operatorname{Tr}\left(X^{(m)} A\right), \qquad \forall A \in \mathcal{A}.
$$
This implies
\begin{equation*}
\label{eq:emrepr}
X^{(m)} = \frac{1}{2}(\check{\rho}_\theta X^{(e)} +X^{(e)} \check{\rho}_\theta),\qquad \forall X \in T_{\check{\rho}_\theta}\check{\mathcal{M}},
\end{equation*}
from which it is easy to obtain $\check{\partial}_j^{(e)}=A_j$. Thanks to the $e$-representation, we introduce an inner product $\langle$,$\rangle$ on $T_{\check{\rho}_\theta}\check{\mathcal{M}}$ by
$$
\begin{aligned}
\langle X, Y\rangle_{\check{\rho}_\theta} & =\ll X^{(e)}, Y^{(e)} \gg_{\check{\rho}_\theta} \\
& =\operatorname{Tr}\left(X^{(m)} Y^{(e)}\right), \qquad \forall X, Y \in T_{\check{\rho}_\theta}\check{\mathcal{M}},
\end{aligned}
$$
which is a quantum version of the Fisher metric \cite{petz2011introduction}. Every element of the quantum Fisher metric is expressed as a real function of $\theta$
$$
g_{k j}(\theta)=\left\langle\check{\partial}_k, \check{\partial}_j\right\rangle_{\check{\rho}_\theta}=\operatorname{Tr}\left(\check{\rho}_\theta A_k A_j\right).
$$
The quantum Fisher information matrix is a real matrix of dimensions $m \times m$ and can be expressed as $G(\theta)=\left(g_{k j}(\theta)\right)$. 
A projection operation that is orthogonal, known as $\Pi_{\theta}$, can be defined to map vectors in $\mathcal{A}$ to vectors in $T_{\check{\rho}_\theta}\check{\mathcal{M}}$ in the following manner:
\begin{equation}
\begin{aligned}
\Pi_{\theta}: \mathcal{A} & \longrightarrow T_{\check{\rho}_\theta}\check{\mathcal{M}} \\
\nu & \longmapsto \sum_{k=1}^{m} \sum_{j=1}^m g^{k j}(\theta)\langle\nu, \check{\partial}_{j}^{(e)}\rangle_{\check{\rho}_\theta} \check{\partial}_k, \label{projection operation}
\end{aligned}
\end{equation}
where the matrix $\left(g^{k j}(\theta)\right)$ refers to the inverse of the quantum information matrix $G(\theta)$.\\
\section{Some basic tools for stochastic processes and stability}\label{app:basic}
Here we recall some basic elements related to definitions of infinitesimal generator, Stratonovich equation, and stochastic stability. 
\subsection*{Infinitesimal generator and It\^{o} formula}
Consider a stochastic differential equation $d q_t=$ $f\left(q_t\right) d t+g\left(q_t\right) d W_t$, where $q_t$ takes  values in $Q \subset \mathbb{R}^p$.  The corresponding infinitesimal generator $\mathscr{L}$ operating on { twice continuously differentiable function in space and continuously differentiable function in time such that}  $V: Q \times \mathbb{R}_{+} \rightarrow \mathbb{R}$ is defined as follows
$$
\mathscr{L} V(q, t):=\frac{\partial V(q, t)}{\partial t}+\sum_{j=1}^p \frac{\partial V(q, t)}{\partial q_j} f_j(q)+\frac{1}{2} \sum_{j, k=1}^p \frac{\partial^2 V(q, t)}{\partial q_j \partial q_k} g_j(q) g_k(q) .
$$
It\^{o} formula gives
$
d V(q, t)=\mathscr{L} V(q, t) d t+\sum_{j=1}^p \frac{\partial V(q, t)}{\partial q_j} g_j(q) d W_t.
$
\subsection*{Stratonovich equation} Any stochastic differential equation in It\^{o} form in $\mathbb{R}^p$
\[d q_t = \tilde{f}\left(q_t\right) d t+
\tilde{g}\left(q_t\right)  d W_t, \quad q_0=q,\]
can be written in the following Stratonovich form 
$$
d q_t = f\left(q_t\right) d t+
g\left(q_t\right) \circ d W_t, \quad q_0=q,
$$
where 
$f(q)=\widetilde{f}(q)-\frac{1}{2} \sum_{l=1}^p  \frac{\partial \tilde{g}}{\partial x_l}(x)\left(\tilde{g}\right)_l(x),\left(\widetilde{g}\right)_l$ 
denoting the $l$-th component  of the vector $\widetilde{g}$, and $g(q)= \widetilde{g}(q)$.
\subsection*{Stochastic stability}
In this section we recall some stability notions for stochastic dynamics. Consider a stochastic process $q_t$
solution of the stochastic differential equation
\[
d q_t = f\left(q_t\right) d t+
g\left(q_t\right)  d W_t, \quad q_0=q,
\]
where $q_t$ takes values on a manifold $\mathcal{N}$ endowed with a metric structure $d$. We assume that the classical existence and uniqueness conditions for the solution are satisfied. 
We say that $\hat q$ is an equilibrium of the system if $f(\hat q)=g(\hat q)=0$.
{ 
We have the following definition.
\begin{definition}[see \cite{Kham}] 
\label{stability}
Let $E$ be a set of  equilibria of the system. Then we say that
\begin{itemize}
\item
$E$ is locally stable in probability, if for every $\varepsilon \in(0,1)$ and for every $r>0$, there exists $\delta=\delta(\varepsilon, r)$ such that
\[\mathbb{P}(d(q_t,E)\leq r\mbox{ for every }t\geq 0)\geq 1-\varepsilon\]
whenever $d(q_0,E)\leq \delta,$
\item  $E$ is almost surely asymptotically stable in $\Gamma$, where $\Gamma \subset \mathcal{N}$ is a.s. invariant, if it is locally stable in probability and 
\[\mathbb{P}(\lim_{t\to\infty} d(q_t,E)=0)=1\]
whenever
$q_0\in\Gamma,$
\item $E$ is almost surely exponentially stable in $\Gamma$, where $\Gamma \subset \mathcal{N}$ is a.s. invariant, if 
\[\limsup _{t \rightarrow \infty} \frac{1}{t} \log d(q_t,E)<0\]
almost surely whenever $q_0\in\Gamma.$ The left-hand side of the inequality is referred to as the sample Lyapunov exponent of the solution,

\item $E$ is exponentially stable in mean in $\Gamma$, where $\Gamma \subset \mathcal{N}$ is a.s. invariant, if \[ \mathbb E (d(q_t,E)) \leq \alpha d(q_0,E)e^{- \beta t}\]
for some positive constants $\alpha$ and $\beta$ whenever $q_0 \in \Gamma$. The smallest value $- \beta$ for which the above inequality is satisfied is called the average Lyapunov exponent.
\end{itemize}
\end{definition}
}


\bibliographystyle{siamplain}
\bibliography{bibliobrahim}
\end{document}